\documentclass[preprint,authoryear,12pt]{elsarticle}
\everymath{\displaystyle}
\usepackage{graphicx}
\usepackage{epsfig}
\usepackage{float}
\usepackage{amsmath}
\usepackage{amssymb}
\usepackage{amsthm}
\usepackage{subfigure}
\newtheorem{define}{Definition}
\newtheorem{thm}{Theorem}
\newtheorem{lemma}{Lemma}
\newtheorem{remark}{Remark}
\newtheorem{assump}{Assumption}
\newtheorem{prop}{Proposition}
\begin{document}
\begin{frontmatter}
\title{Adaptive-Gain Second-Order Sliding Mode Observer Design for Switching Power Converters}
\author{Jianxing Liu}
\author{S.Laghrouche}
\author{M.Harmouche}
\author{M.Wack}
\address{Laboratoire IRTES, Universit\'e de Technologie de Belfort-Montb\'eliard, Belfort, France.\\
(e-mail: Jiang-Xing.Liu@utbm.fr; salah.laghrouche@utbm.fr; Mohamed.Harmouche@utbm.fr; maxime.wack@utbm.fr)}
\begin{abstract}                
In this paper, a novel adaptive-gain,
Second Order Sliding Mode (SOSM) observer
for multi-cell converters is designed by considering it as a type
of hybrid system.
The objective is to reduce the number of voltage sensors by
estimating the capacitor voltages from only measurements of the load current.
The proposed observer is proven to be robust in
the presence of perturbations with \emph{unknown} boundaries.
However, the states of the system are only partially observable
based on the observability matrix rank condition.
Because its observability depends upon the switching control signals,
a recent concept known as $Z(T_N)$-observability,
which can be used to analyze the observability of hybrid systems,
is used to address the switching behavior.
Under certain conditions of the switching sequences,
the voltage across each capacitor becomes observable.
Simulation results and comparisons with a Luenberger switched observer
demonstrate the effectiveness and the robustness of the proposed observer
with respect to output measurement noise and system uncertainties (load variations).
\end{abstract}
\end{frontmatter}
\section{INTRODUCTION}
In recent years, industrial applications requiring high power levels have used medium-voltage semiconductors \citep{meynard1992multi,rodriguez2002multilevel,rech2007hybrid,gerry2003high}.
 Because of the efficiency requirements, the power of the converter is generally increased by
boosting the voltage.
However, medium-voltage switching devices are not available.
Even if they did exist, the volume and the cost of such devices would
be substantial \citep{1036749}.
In this sense, the topology of multi-level converters,
which have been studied during the last decade,
becomes attractive for high voltage applications \citep{meynard1992multi}.
From a practical point of view, the series of a multi-cell chopper
designed by the LEEI (Toulouse, France) \citep{bensaid},
leads to a safe series association of components working in a switching mode.
This structure offers the possibility of reducing the voltage constraints
evenly among each cell in a series.
These lower-voltage switches result in
lower conduction losses and higher switching frequencies.
Moreover, it is possible to improve the output waveforms
using this structure \citep{954041,1036749,1023142}.
These flying capacitors have to be balanced to guarantee the desired voltage
values at the output, which ensures that the maximum
benefit from the multi-cell structure is obtained \citep{meynard1997modeling}.
These properties are lost if the
capacitor voltage drifts far from the desired value \citep{bejaranomulticell}.
Therefore, a suitable control of the switches is required to generate
the desired values of the capacitor voltages.
The control of switches allows the current harmonics
at the cutting frequency to be canceled and
the ripple of the chopped voltage to be reduced \citep{djemai2011high,Defoortmulticell}.

Several control methods have been proposed for multi-cell converters,
such as nonlinear control based on input-output linearization \citep{1036749},
predictive control \citep{defay2008predictive}, hybrid control \citep{baja2007hybrid},
model predictive control \citep{defay2008predictive,lezana2009model}
and sliding mode control
\citep{djemai2011high,5991333,Meradi2013128}.
However, most of these techniques require measurements of the voltages of the capacitors to design the controller. That is, extra voltage sensors are necessary, which increases the
cost and the complexity of the system. Hence, the estimation of the capacitor voltages using an observer has attracted great interest \citep{besancon2007nonlinear}.

It should be noted that the states of the multi-cell system are only partially observable because the observability matrix never has full rank \citep{besancon2007nonlinear}.
Hence, the observability matrix rank condition cannot be employed
in an observability analysis of a hybrid system
such as the one considered here \citep{vidal2003observability,babaali2005observability}.
A recent concept, $Z(T_N)$-observability \citep{KANG}, can be
used to analyze the observability of a switched hybrid system and is applied in this work because the observability of the converter depends upon the switching control signals.
Various observers have been designed for the multi-cell converters based on
concepts such as
homogeneous finite-time observers \citep{Defoortmulticell},
super-twisting sliding mode observers \citep{bejaranomulticell,ghanes2009sliding}
and adaptive observer \citep{bejaranomulticell}.
The concept of observability presented in \cite{KANG} gives the condition under which there exists
a hybrid time trajectory that makes the system observable.
Using this concept,
estimates of the capacitor voltages can be obtained from the measurements of
the load current and the source voltage by taking advantage of the appropriate
hybrid time trajectories.

In this paper,
an observability analysis based on the results of \cite{KANG}
is performed for the multi-cell converter assuming measurements of the load
current and the source voltage under certain conditions of the switching input sequences.
Then, a novel adaptive-gain SOSM observer for multi-cell converters
is introduced that takes into account certain perturbations (load variations)
in which the boundaries of their first time derivatives are \emph{unknown}.
The proposed adaptive-gain SOSM algorithm combines the nonlinear term of
the super-twisting algorithm and a linear term, the so-called SOSML algorithm \citep{moreno2008lyapunov}.
The behavior of the SOSM algorithm near the origin is significantly improved compared with the linear case.
Conversely, the additional linear term improves the behavior of the SOSM algorithm when the states are far from the origin.
Therefore, the SOSML algorithm inherits the best properties
of both the linear and the nonlinear terms.
An adaptive law of the gains of the SOSML algorithm is derived
via the so-called "time scaling" approach \citep{Respondek2004277}.
The output observation error and its first time derivative converge to
zero in finite time with the proposed SOSML observer
such that the equivalent output-error injection can be obtained directly.
Finally, the resulting reduced-order system is proven to be exponentially stable.
That is, the observer error for the capacitor voltages, which are considered as the states of the observer system, converge to zero exponentially.
The main advantages of the proposed adaptive-gain SOSML are that
only one parameter has to be tuned
and there are no a-priori requirements on the perturbation bounds.

This paper is organized as follows. In Section $\mathrm{II}$,
a model of the multi-cell converter and its characteristics are presented.
In Section $\mathrm{III}$, the observability of the multi-cell converter is studied with
the concept of $Z(T_N)$-observability.
Section $\mathrm{IV}$ discusses the design of the proposed adaptive-gain SOSML
observer for estimating the capacitor voltages.
Section $\mathrm{V}$ gives simulation results including a comparison with a Luenberger switched observer with disturbances.
\section{MODELING OF the MULTI-CELL CONVERTER}\label{sec:modelling}
The structure of a multi-cell converter is based on the combination of
a certain number of cells.
Each cell consists of an energy storage element and commutators \citep{1036749}.
The main advantage of this structure is that the spectral quality of the output signal is
improved by a high switching frequency between the intermediate
voltage levels \citep{mcgrath2007analytical}.
An instantaneous model that was presented in \cite{1036749} and
describes fully the hybrid behavior of the multi-cell converter
is used here.

Figure \ref{fig:multicell_model} depicts the topology of a converter with $p$ independent commutation cells that is connected to an inductive load.
The current $I$ flows from the source $E$ to the output
through the various converter switches.
The converter thus has a hybrid behavior because of the presence of both discrete variables (the switching logic)
and continuous variables (the currents and the voltages).
\begin{figure}[H]
  \begin{center}
  \includegraphics[width=8.4cm]{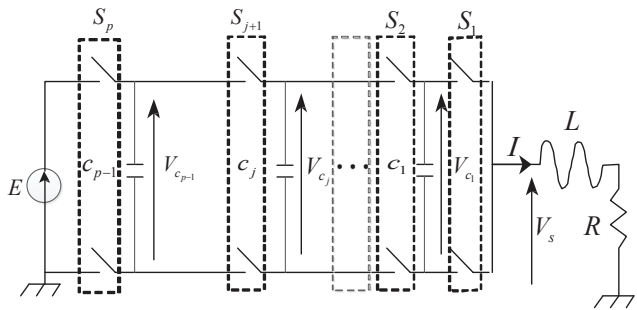}
  \caption{Multicell converter on RL load}
  \label{fig:multicell_model}
  \end{center}
\end{figure}

Through circuit analysis, the dynamics of the p-cell converter
were obtained as the following differential equations:
\begin{equation}\label{eqn:dynamics of converter}
\left\{
\begin{split}
\dot{I}&=-\frac{R}{L}I + \frac{E}{L}S_p - \sum_{j=1}^{p-1}\frac{V_{c_j}}{L}(S_{j+1}-S_j),\\
\dot{V}_{c_1}&=\frac{I}{c_1}(S_{2}-S_1),\\
&\vdots\\
\dot{V}_{c_{p-1}}&=\frac{I}{c_p}(S_{p}-S_{p-1}).
\end{split}
\right.
\end{equation}
where $I$ is the load current, $c_j$ is the $j^{th}$ capacitor, $V_{c_j}$ is the voltage of the $j^{th}$ capacitor
and $E$ is the voltage of the source. Each commutation cell is controlled by the binary input signal $S_j\in\{0,1\}$, where
$S_j=1$ indicates that the upper switch of the $j$th cell is on and the lower switch is off
and $S_j=0$ indicates that the upper switch is off and the lower switch is on.
The discrete inputs are defined as follows:
\begin{equation}\label{eqn:control input}
\left\{
\begin{split}
u_j=&S_{j+1}-S_j,\quad j=1,\ldots,p-1\\
u_p=&S_p.
\end{split}
\right.
\end{equation}

With equation (\ref{eqn:control input}),
the system (\ref{eqn:dynamics of converter}) can be represented as follows:
\begin{equation}\label{eqn:system dynamic}
\left\{
\begin{split}
&\dot{I}=-\frac{R}{L}I + \frac{E}{L}u_p - \sum_{j=1}^{p-1}\frac{V_{c_j}}{L}u_j,\\
&\dot{V}_{c_1}=\frac{I}{c_1}u_1,\\
&\vdots\\
&\dot{V}_{c_{p-1}}=\frac{I}{c_{p-1}}u_{p-1},\\
&y=I.
\end{split}
\right.
\end{equation}

Assuming that only the load current $I$ can be measured,
it is easy to represent the system (\ref{eqn:system dynamic}) as
a hybrid (switched affine) system:
\begin{equation}\label{eqn:hybrid dynamic1}
\left\{
\begin{split}
\dot{x} =& f(x,u)=A(u)x + B(u),\\
y=&h(x,u)=Cx.
\end{split}
\right.
\end{equation}
where
$x=
\begin{bmatrix}
    I&V_{c_1}&\cdots&V_{c_{p-1}}
\end{bmatrix}
^T$
is the continuous state vector,\\
$u=
\begin{bmatrix}
    u_1&u_2&\cdots&u_p
\end{bmatrix}
^T$
is the switching control signal vector which takes only discrete values
and the matrices $A(u)$, $B(u)$, $C$ are defined as:
\begin{equation}
\begin{split}
A(u)&=
\begin{bmatrix}
    -\frac{R}{L}&-\frac{u_1}{L}&\cdots&-\frac{u_{p-1}}{L}\\
    \frac{u_1}{c_1}&0&\cdots&0\\
    \vdots&\vdots&\ddots&\vdots\\
    \frac{u_{p-1}}{c_{p-1}}&0&\cdots&0
\end{bmatrix}
,\\
B(u)&=
\begin{bmatrix}
    \frac{E}{L}u_p&0&\cdots&0
\end{bmatrix}
^T,
\\
C&=
\begin{bmatrix}
    1&0&\cdots&0
\end{bmatrix}
.
\end{split}
\end{equation}

The main objective of this paper is to design an observer based on the instantaneous model (\ref{eqn:system dynamic}) that is able to estimate
the capacitor voltages using only
the measurement of the load current and the associated switching control input (which is assumed to be known).
\section{HYBRID OBSERVABILITY ANALYSIS}\label{sec:Observability Analysis}
From the instantaneous model of the system (\ref{eqn:hybrid dynamic1}) with $p\geq 3$,
it can be noted that there are several switching modes
that make the system unobservable.
For instance, if $u_1=u_2=\cdots=u_{p-1}=0$, the voltages $V_{c_j}(j={1,\ldots,p-1})$
become completely unobservable.
These switching modes are not affected by the capacitor voltages.
Fortunately, these cases are ones in which the $p$-cells are not switching and will not occur for all control sequences;
otherwise, there is no interest in the physical sense.

The observability analysis of the system (\ref{eqn:hybrid dynamic1}) is based
on the measurement of the load current $I$ and the knowledge of the control input sequence $u$.
The so-called observability matrix \citep{besancon2007nonlinear} is
defined as
\begin{equation}\label{matrix:observability}
\begin{split}
\mathcal{O}_{p\times p}&=
\begin{bmatrix}
    C\\CA\\CA^2\\\vdots\\CA^{p-1}
\end{bmatrix}
=
\begin{bmatrix}
    1&0&\cdots&0\\
    -\frac{R}{L}&-\frac{u_1}{L}&\cdots&-\frac{u_{p-1}}{L}\\
    (\frac{R}{L})^2-\sum^{p-1}_{i=1}\frac{u^2_i}{Lc_i}&\frac{Ru_1}{L^2}&\cdots&\frac{Ru_{p-1}}{L^2}\\
    \vdots&\vdots&\vdots&\vdots\\
\end{bmatrix}
.
\end{split}
\end{equation}
With simple computations, it can be shown that
\begin{equation}\label{matrix:observability threecells}
\begin{split}
rank(\mathcal{O})&=2<p.
\end{split}
\end{equation}
It follows that the continuous states are not observable using
only the load current because
the observability matrix (\ref{matrix:observability threecells})
is not full rank.

Because of the switching sequences of the system (\ref{eqn:hybrid dynamic1}),
the observability is strongly linked to the hybrid behavior.
Therefore, the recently developed concept of $Z(T_N)$-observability \citep{KANG}
is applied to analyze the observability of the hybrid system (\ref{eqn:hybrid dynamic1}).
It is important to note the following definitions.
\begin{define}\cite{KANG}\label{define:hybrid time trajectory}
A hybrid time trajectory is a finite or infinite sequence of intervals $T_N={\Gamma}^N_{i=0}$ such that
\begin{itemize}
  \item $\Gamma_i=[t_{i,0},t_{i,1}), for $\ $ all $\ $0\leq i<N$;
  \item For all $i<N, $\ $t_{i,1}=t_{i+1,0}$;
  \item $t_{0,0}=t_{ini} $\ $ and $\ $ t_{N,1}=t_{end}$.
\end{itemize}
Moreover, $\langle T_N\rangle$ is defined as the ordered list of inputs $u$ associated with $T_N$, ${u^i}_{i=0,N}$,
where $u^i$ is the value of $u$ on the interval $\Gamma_i$.
\end{define}
\begin{define}\citep{KANG}\label{define:2}
The function $z=Z(t,x)$ is said to be $Z-$observable with respect to the hybrid time trajectory $T_N$
and $\langle T_N\rangle$ if for any two trajectories $(t,x,u)$ and $(t,x',u')$
defined in $[t_{ini},t_{end}]$, the equality $h(x,u)=h(x',u')$ implies that $Z(t,x)=Z(t,x')$.
\end{define}
\begin{lemma}\citep{KANG}\label{define:Zobservability}
Consider the system (\ref{eqn:hybrid dynamic1}) and a fixed hybrid time trajectory $T_N$ and $\langle T_N\rangle$.
Suppose that $z=Z(t,x)$ is always continuous under any admissible control input.
If there exists a sequence of projections ${P_i},$\ $i=0,1,\cdots,N$, such that
\begin{itemize}
  \item For all $i<N$, $P_iZ(t,x)$ is $Z-$observable for $t\in \Gamma_i$;
  \item Rank$([P^T_0,\cdots,P^T_N])=dim(z)$;
  \item $\frac{d\bar{P}_iZ(t,x)}{dt}=0$, for $t\in \Gamma_i$, where $\bar{P}$ is the complement
  of $P$(projecting $z$ to the variables eliminated by $P$),
\end{itemize}
then, $z=Z(t,x)$ is $Z-$observable with respect to the hybrid time trajectory $T_N$ and $\langle T_N\rangle$.
\end{lemma}
\begin{proof}
The proof of the Lemma \ref{define:Zobservability}
can be found in \cite{KANG}.
\end{proof}
\begin{remark}\label{remark:1}
In Lemma \ref{define:Zobservability}, the third condition requires that the components of $Z$
that are not observable in $\Gamma_i$ must remain constant within this time interval.
The hybrid time trajectory $T_N$ and $\langle T_N\rangle$
influences the observability property in a way similar
to an input.
\end{remark}

Table \ref{Different modes} presents the eight possible configurations for a three-cell converter. The application of Lemma \ref{define:Zobservability} to the three-cell converter is as follows.
We take $Z(t,x)=\begin{bmatrix}x_2&x_3\end{bmatrix}^T=\begin{bmatrix}V_{c_1}&V_{c_2}\end{bmatrix}^T$.
For the discrete switching conditions $[0,0,0]$ and $[1,1,1]$, it can be verified that $Z(t,x)$ is
not $Z(T_N)$-observable. Fortunately, from (\ref{eqn:dynamics of converter}) the dynamics of $V_{c_1}$ and $V_{c_2}$ are zero, which means that
these states remain constant during these time intervals.
Next, assume that a trajectory of the system has the status $[1,0,0]$ and $[1,1,0]$ during time intervals $\Gamma_1$ and $\Gamma_2$, respectively.
Let us define $P_1=\begin{bmatrix}1&0\end{bmatrix}$ and $P_2=\begin{bmatrix}0&1\end{bmatrix}$.
We have $\bar{P}_1Z=x_3=V_{c_2}$, $\frac{d\bar{P}_1Z}{dt}=\frac{dV_{c_2}}{dt}=0$,
$\bar{P}_2Z=x_2=V_{c_1}$, $\frac{d\bar{P}_2Z}{dt}=\frac{dV_{c_1}}{dt}=0$
and $rank\begin{bmatrix}P_1\\P_2\end{bmatrix}=2$.
All the assumptions in Lemma \ref{define:Zobservability} are satisfied;
therefore, $Z(t,x)=\begin{bmatrix}V_{c_1}&V_{c_2}\end{bmatrix}^T$ is $Z(T_N)$-observable.
\begin{table}[h]
\begin{center}
\caption{Switching modes and capacitor voltages for a three-cell converter}
\label{Different modes}
\begin{tabular}{cccccl}
\hline
Mode : $[S_1,S_2,S_3]$ & $V_{c_1}$&$V_{c_2}$&$u_1$&$u_2$&Observable States\\
\hline
0 : [0,0,0] &$\rightsquigarrow$&$\rightsquigarrow$&0&0&$I$\\
\hline
1 : [0,0,1] &$\rightsquigarrow$&$\nearrow$&0&1&$I,V_{c_2}$\\
\hline
2 : [0,1,0] &$\nearrow$&$\searrow$&1&-1&$I,V_{c_1},V_{c_2}$\\
\hline
3 : [0,1,1] &$\nearrow$&$\rightsquigarrow$&1&0&$I,V_{c_1}$\\
\hline
4 : [1,0,0] &$\searrow$&$\rightsquigarrow$&-1&0&$I,V_{c_1}$\\
\hline
5 : [1,0,1] &$\searrow$&$\nearrow$&-1&1&$I,V_{c_1},V_{c_2}$\\
\hline
6 : [1,1,0] &$\rightsquigarrow$&$\searrow$&0&-1&$I,V_{c_2}$\\
\hline
7 : [1,1,1] &$\rightsquigarrow$&$\rightsquigarrow$&0&0&$I$\\
\hline
\end{tabular}
\end{center}
\end{table}

The symbols in Table \ref{Different modes} are defined as follows:
$\rightsquigarrow$ indicates a constant value, $\nearrow$ indicates increasing and $\searrow$ indicates decreasing.

In the next section, an adaptive-gain SOSML observer will
be presented for the converter system (\ref{eqn:system dynamic}).
\section{ADAPTIVE-GAIN SOSML OBSERVER DESIGN}
As discussed in \cite{1036749}, active control of the multi-cell converter
requires the knowledge of the capacitor voltages.
Usually, voltage sensors are used to measure the capacitor voltages.
However, the extra sensors increase the cost, the complexity and the size,
especially in high-voltage applications.
Moreover, any sensors will introduce the measurement noise which
will be directly transposed to the estimated value.
Therefore, the design of a state observer using only the measurement of load current and the associated switching inputs is desirable.

In this section, an adaptive-gain SOSML observer
for the three-cell converter ($p=3$) is presented that is robust to perturbations (load variations) for which the boundaries of the first time derivative are \emph{unknown}.
A novel adaptive law for the gains of the SOSML algorithm
with only one tuning parameter is designed via
the so-called "time scaling" approach \citep{Respondek2004277}.
The proposed approach does not require the a-priori knowledge of the perturbation bounds.

Defining $e_1=I-\hat{I}$,
the system (\ref{eqn:system dynamic}) is rewritten to include
the perturbation $\tilde{f}(e_1)$,
i.e., the load resistance uncertainty \citep{Defoortmulticell},
\begin{equation}\label{eqn:new system dynamic}
\left\{
\begin{split}
&\dot{I}=-\frac{R}{L}I + \frac{E}{L}u_3 - \frac{V_{c_1}}{L}u_1-
\frac{V_{c_2}}{L}u_2 + \tilde{f}(e_1),\\
&\dot{V}_{c_1}=\frac{u_1}{c_1}I,\\
&\dot{V}_{c_2}=\frac{u_2}{c_2}I.
\end{split}
\right.
\end{equation}
The proposed observer is formulated as
\begin{equation}\label{eqn:stw observer}
\left\{
\begin{split}
&\dot{\hat{I}}=-\frac{R}{L}I + \frac{E}{L}u_3 - \frac{\hat{V}_{c_1}}{L}u_1-
\frac{\hat{V}_{c_2}}{L}u_2 + \mu(e_1),\\
&\dot{\hat{V}}_{c_1}=\frac{u_1}{c_1}I + k_1\mu(e_1),\\
&\dot{\hat{V}}_{c_2}=\frac{u_2}{c_2}I + k_2\mu(e_1).
\end{split}
\right.
\end{equation}
where $\mu(\cdot)$ is the SOSML algorithm \citep{moreno2008lyapunov},
\begin{equation}\label{eqn:stw definition}
\mu(e_1)=\lambda(t) |e_1|^{\frac{1}{2}}sign(e_1)+\alpha(t)
\int^t_0 sign(e_1)d\tau+k_\lambda(t)e_1+k_\alpha(t)\int^t_0 e_1d\tau,
\end{equation}
and the adaptive gains
$\lambda(t),\alpha(t),k_\lambda(t),k_\alpha(t)$ and
the design parameters $k_1$ and $k_2$
are to be defined.

Define the observation errors as
\begin{equation}\label{eqn:error define}
\left\{
\begin{split}
e_2=&V_{c_1}-\hat{V}_{c_1},\\
e_3=&V_{c_2}-\hat{V}_{c_2}.
\end{split}
\right.
\end{equation}
Equations (\ref{eqn:new system dynamic}) and (\ref{eqn:stw observer})
yield the observation error dynamics as:
\begin{eqnarray}
\dot{e}_1 &=& -\mu(e_1) -\frac{u_1}{L}e_2 - \frac{u_2}{L}e_3 + \tilde{f}(e_1),
\label{eqn:dynamicobservererror1}\\
\dot{e}_2 &=& -k_1\mu(e_1),
\label{eqn:dynamicobservererror2}\\
\dot{e}_3 &=& -k_2\mu(e_1).
\label{eqn:dynamicobservererror3}
\end{eqnarray}

In this paper, the adaptive gains $\lambda(t),\alpha(t),k_\lambda(t)$ and $k_\alpha(t)$
are formulated as
\begin{equation}\label{eqn:adaptive gains}
\left\{
\begin{split}
\lambda(t)&=\lambda_0\sqrt{l(t)},\\
\alpha(t)&=\alpha_0l(t),\\
k_\lambda(t)&=k_{\lambda_0}l(t),\\
k_\alpha(t)&=k_{\alpha_0}l^2(t).
\end{split}
\right.
\end{equation}
where $\lambda_0,\alpha_0,k_{\lambda_0}$ and $k_{\alpha_0}$ are
positive constants to be defined and $l(t)$ is a positive, time-varying, scalar function.

The adaptive law of the time-varying function $l(t)$ and
the design parameters $k_1$ and $k_2$ are given by:
\begin{equation}\label{eqn:adaptive gains rt}
\dot{l}(t) =
\begin{cases}
  k, \quad\quad if \quad |e_1|\neq0\\
  0, \quad\quad else
\end{cases}
\end{equation}
\begin{equation}\label{ineqn:k1k2}
\begin{split}
k_1=&
\begin{cases}
  -\kappa u_1, \ \ if \ \ |e_1|=0\\
  0. \ \ \quad\quad else
\end{cases}
,\quad
k_2=
\begin{cases}
  -\kappa u_2, \ \  if \ \ |e_1|=0\\
  0. \ \ \quad\quad else
\end{cases}
\end{split}
\end{equation}
where $k$, the initial value $l(0)$  and $\kappa$ are positive constants.
\begin{assump}\label{assump:1}
The system (\ref{eqn:new system dynamic}) and the observer system
(\ref{eqn:stw observer}) are bounded input, bounded state (BIBS)
because this is a physical system \rm{\citep{perruquetti2002sliding}}.
\end{assump}
\begin{assump}\label{assump:5}
There is a $T_N$ such that $z=x=\begin{bmatrix}I&V_{c_1}&V_{c_2}\end{bmatrix}^T$
is Z-observable under the condition
of Lemma \ref{define:Zobservability} \rm{\citep{Defoortmulticell}}.
\end{assump}
\begin{thm}\label{theorem:5}
Consider the error system (\ref{eqn:dynamicobservererror1}) under Assumptions
(\ref{assump:1}) and (\ref{assump:5}).
Assume that the perturbation $\tilde{f}(e_1)$ satisfies the following condition:
\begin{equation}\label{condition1}
\begin{split}
\left\vert\dot{\tilde{f}}(e_1)\right\vert\leq \chi_1,\ \ and \ \
\tilde{f}(0)=0.
\end{split}
\end{equation}
where $\chi_1$ is an \emph{unknown} positive constant.
Then, the trajectories of the error system (\ref{eqn:dynamicobservererror1})
converge to zero in finite time with the adaptive gains in
(\ref{eqn:adaptive gains}) and (\ref{eqn:adaptive gains rt}) satisfying the following condition:
\begin{equation}\label{ineqn:condition1}
    4\alpha_0k_{\alpha_0}>8k^2_{\lambda_0}\alpha_0+9\lambda^2_0k^2_{\lambda_0},
\end{equation}
\end{thm}
\begin{proof}\label{pf:adaptive sliding}
The system (\ref{eqn:dynamicobservererror1}) can be rewritten as
\begin{equation}\label{eqn:two intergrals}
\left\{
\begin{split}
  \dot{e}_1 & =-\lambda(t)|e_1|^{\frac{1}{2}}sign(e_1)-k_\lambda(t)e_1+\varphi_1,\\
  \dot{\varphi}_1 & =-\alpha(t)sign(e_1) -k_\alpha(t)e_1+ \varrho_1.
\end{split}
\right.
\end{equation}
where $\varrho_1=\left(\dot{\tilde{f}}(e_1)-\frac{u_1}{L}\dot{e}_2 -
\frac{u_2}{L}\dot{e}_3\right)$.

Based on Assumption \ref{assump:1}, because the input $u$ is bounded,
the state does not go to infinity in finite time. Moreover, if $\hat{I}$ is bounded, all the states of the observer are also bounded for a finite time. Consequently,
the observation error $e_1$ is also bounded \citep{perruquetti2002sliding}.
It follows from (\ref{eqn:dynamicobservererror2}, \ref{eqn:dynamicobservererror3})
that $\dot{e}_2$ and $\dot{e}_3$ are bounded and satisfy
$|\dot{e}_2|\leq\chi_2$ and $|\dot{e}_3|\leq\chi_3$,
where $\chi_2$ and $\chi_3$ are some \emph{unknown} positive values.
From equation (\ref{condition1}), it is easy to deduce that
$|\varrho_1|\leq\chi_1+\frac{\chi_2}{L}+\frac{\chi_3}{L}=F$, where
$F$ is an \emph{unknown} positive value.

A new state vector is introduced to represent the system in (\ref{eqn:two intergrals})
in a more convenient form for Lyapunov analysis.
\begin{equation}\label{eqn:new state vector}
  \zeta = \begin{bmatrix}
    \zeta_1\\ \zeta_2\\ \zeta_3
    \end{bmatrix}
    =\begin{bmatrix}
    l^{\frac{1}{2}}(t)|e_1|^{\frac{1}{2}}sign(e_1)\\
    l(t)e_1\\ \varphi_1
    \end{bmatrix},
\end{equation}
Thus, the system in (\ref{eqn:two intergrals}) can be rewritten as
\begin{equation}\label{eqn:new vector dynamics}
\begin{split}
  \dot{\zeta} & =\frac{l(t)}{|\zeta_{1}|}
\underbrace{\begin{bmatrix}
    -\frac{\lambda_0}{2}&0&\frac{1}{2}\\0&-\lambda_0&0\\
    -\alpha_0&0&0
\end{bmatrix}}_{A_1}\zeta+l(t)
\underbrace{\begin{bmatrix}
    -\frac{k_{\lambda_0}}{2}&0&0\\
    0&-k_{\lambda_0}&1\\
    0&-k_{\alpha_0}&0
\end{bmatrix}}_{A_2}\zeta
+
\begin{bmatrix}
    \frac{\dot{l}}{2l(t)}\zeta_{1}\\ \frac{\dot{l}}{2l(t)}\zeta_{2}\\ \varrho_1
\end{bmatrix}
,
\end{split}
\end{equation}

Then, the following Lyapunov function candidate is introduced for the system
in (\ref{eqn:new vector dynamics}):
\begin{equation}\label{eqn:new lyapunov function1}
    V(\zeta)=2\alpha_0\zeta^2_1+k_{\alpha_0}\zeta^2_2+\frac{1}{2}\zeta^2_3
    +\frac{1}{2}\left(\lambda_0\zeta_1+k_{\lambda_0}\zeta_2-\zeta_3\right)^2,
\end{equation}
which can be rewritten as a quadratic form
\begin{equation}\label{eqn:lyapunov function1}
    V(\zeta)=\zeta^TP\zeta,\quad P =\frac{1}{2}
    \begin{bmatrix}
    4\alpha_0+\lambda^2_0&\lambda_0k_{\lambda_0}&-\lambda_0\\
    \lambda_0k_{\lambda_0}&k^2_{\lambda_0}+2k_{\alpha_0}&-k_{\lambda_0}\\
    -\lambda_0&-k_{\lambda_0}&2
    \end{bmatrix}.
\end{equation}
As (\ref{eqn:new lyapunov function1}) is a continuous Lyapunov function, the matrix $P$ is positive definite.

Taking the derivative of (\ref{eqn:lyapunov function1}) along
the trajectories of (\ref{eqn:new vector dynamics}),
\begin{equation}\label{eqn:derivative lyapunov}
\begin{split}
    \dot{V}=-\frac{l(t)}{|\zeta_{1}|}\zeta^T\Omega_1\zeta-
    l(t)\zeta^T\Omega_2\zeta+\varrho_1q_1\zeta+\frac{\dot{l}}{l(t)}q_2P\zeta,
\end{split}
\end{equation}
where $q_1=\begin{bmatrix}
    -\lambda_0&-k_{\lambda_0}&2
    \end{bmatrix}$,
$q_2=\begin{bmatrix}
    \zeta_1&\zeta_2&0
    \end{bmatrix}$,
and
\begin{equation}\label{eqn:matrix3}
\begin{split}
\Omega_1&=\frac{\lambda_0}{2}
\begin{bmatrix}
\lambda^2_0+2\alpha_0&0&-\lambda_0\\
0&2k_{\alpha_0}+5k^2_{\lambda_0}&-3k_{\lambda_0}\\
-\lambda_0&-3k_{\lambda_0}&1
\end{bmatrix}
,\\
\Omega_2&=k_{\lambda_0}
\begin{bmatrix}
\alpha_0+2\lambda^2_0&0&0\\
0&k_{\alpha_0}+k^2_{\lambda_0}&-k_{\lambda_0}\\
0&-k_{\lambda_0}&1
\end{bmatrix}.
\end{split}
\end{equation}
it is easy to verify that $\Omega_1$ and $\Omega_2$ are positive definite matrices under the condition in (\ref{ineqn:condition1}).

Because $\lambda_{min}(P)\|\zeta\|^2\leq V\leq \lambda_{max}(P)\|\zeta\|^2$,
Equation (\ref{eqn:derivative lyapunov}) can be rewritten as
\begin{equation}\label{eqn:derivative lyapunov1}
\begin{split}
    \dot{V}\leq-l(t)\frac{\lambda_{min}(\Omega_1)}{\lambda^{\frac{1}{2}}_{max}(P)}V^{\frac{1}{2}}
-l(t)\frac{\lambda_{min}(\Omega_2)}{\lambda_{max}(P)}V
+\frac{F\|q_1\|_2}{\lambda^{\frac{1}{2}}_{min}(P)}V^{\frac{1}{2}}
+\frac{\dot{l}}{2l(t)}\Delta\Omega
\end{split},
\end{equation}
where
\begin{equation}\label{eqn:deltaV}
\begin{split}
\Delta \Omega =&\left(
(4\alpha_0+\lambda^2_0)\zeta^2_1+2\lambda_0k_{\lambda_0}\zeta_1\zeta_2
+2k_{\alpha_0}k^2_{\lambda_0}\zeta^2_2-\lambda_0\zeta_1\zeta_3
-k_{\lambda_0}\zeta_2\zeta_3
\right)\leq\zeta^TQ\zeta,\\
Q=&
\begin{bmatrix}
    4\alpha_0+\lambda^2_0+\lambda_0k_{\lambda_0}+\frac{\lambda_0}{2}&0&0\\
    0&2k_{\alpha_0}k^2_{\lambda_0}+\lambda_0k_{\lambda_0}+\frac{k_{\lambda_0}}{2}&0\\
    0&0&\frac{\lambda_0+k_{\lambda_0}}{2}
\end{bmatrix}.
\end{split}
\end{equation}
With (\ref{eqn:deltaV}), equation (\ref{eqn:derivative lyapunov1}) becomes
\begin{equation}\label{eqn:derivative lyapunov2}
\begin{split}
\dot{V}\leq-\left(
l(t)\frac{\lambda_{min}(\Omega_1)}{\lambda^{\frac{1}{2}}_{max}(P)}
-\frac{F\|q_1\|_2}{\lambda^{\frac{1}{2}}_{min}(P)}
\right)V^{\frac{1}{2}}
-\left(l(t)\frac{\lambda_{min}(\Omega_2)}{\lambda_{max}(P)}-
\frac{\dot{l}}{2l(t)}\frac{\lambda_{max}(Q)}{\lambda_{min}(P)}
\right)V,
\end{split}
\end{equation}
For simplicity, we define
\begin{equation}\label{eqn:parameter definition}
\begin{split}
\gamma_1&=\frac{\lambda_{min}(Q)}{\lambda^{\frac{1}{2}}_{max}(P)},\ \
\gamma_2=\frac{F\|q_1\|_2}{\lambda^{\frac{1}{2}}_{min}(P)},\ \
\gamma_3=\frac{\lambda_{min}(\Omega_2)}{\lambda_{max}(P)},\ \
\gamma_4=\frac{\lambda_{max}(Q)}{2\lambda_{min}(P)}.
\end{split}
\end{equation}
where $\gamma_1,\gamma_2,\gamma_3$ and $\gamma_4$ are all positive constants.
Thus, equation (\ref{eqn:derivative lyapunov2}) can be simplified as
\begin{equation}\label{eqn:derivative lyapunov3}
\begin{split}
    \dot{V}\leq-\left(l(t)\gamma_1-\gamma_2\right)V^{\frac{1}{2}}
-\left(l(t)\gamma_3-\frac{\dot{l}}{l(t)}\gamma_4\right)V,
\end{split}
\end{equation}
Because $\dot{l}(t)\geq 0$ such that the terms $l(t)\gamma_1-\gamma_2$ and
$l(t)\gamma_3-\frac{\dot{l}}{l(t)}\gamma_4$ are positive in finite time,
it follows from (\ref{eqn:derivative lyapunov3}) that
\begin{equation}\label{eqn:final lyapunov}
 \dot{V}\leq -c_1V^{\frac{1}{2}}-c_2V.
\end{equation}
where $c_1$ and $c_2$ are positive constants.
By the comparison principle \citep{khalil1996nonlinear}, it follows
that $V(\zeta)$ and therefore $\zeta$ converge to zero in finite time.
Thus, Theorem \ref{theorem:5} is proven.
\end{proof}
It follows from Theorem \ref{theorem:5} that
when the sliding motion takes place, $e_1=0$ and $\dot{e}_1=0$.
Thus, the output-error equivalent injection $\mu(e_1)$ can be obtained directly from equation (\ref{eqn:dynamicobservererror1}):
\begin{equation}\label{eqn:equivalent injection}
\mu(e_1)=-\frac{u_1}{L}e_2 - \frac{u_2}{L}e_3,
\end{equation}
Substitute (\ref{eqn:equivalent injection}) into the error system
(\ref{eqn:dynamicobservererror2}) and (\ref{eqn:dynamicobservererror3}),
the following reduced-order system is obtained:
\begin{equation}\label{eqn:reduced order dynamics}
\left\{
\begin{split}
\dot{e}_2 =& k_1(\frac{u_1}{L}e_2 + \frac{u_2}{L}e_3),\\
\dot{e}_3 =& k_2(\frac{u_1}{L}e_2 + \frac{u_2}{L}e_3).
\end{split}
\right.
\end{equation}
\begin{prop}\label{prop:1}
Consider the reduced-order system (\ref{eqn:reduced order dynamics})
with the switching gains $k_1$ and $k_2$ given by (\ref{ineqn:k1k2}).
Then, the trajectories of the error system
(\ref{eqn:reduced order dynamics})  converge to zero exponentially if
the following two conditions are satisfied \rm{\citep{loria2002uniform}}:
\begin{itemize}
  \item There exists a constant $\phi_M>0$ such that for all $t\geq 0$ and all $u\in\mathcal{D}$,
  where $\mathcal{D}\in\mathbb{R}^2$ is a closed, compact subset, such that $\|\Psi(t,u)\|\leq\phi_M$,
where
\begin{equation}\label{conditon:1}
    \Psi(t,u)=
\begin{bmatrix}
    \sqrt{\frac{k}{L}}u_1(t)& \sqrt{\frac{k}{L}}u_2(t)
\end{bmatrix}
^T,
\end{equation}
  \item There exist constants $T_1>0$ and $\mu>0$ such that
\begin{equation}\label{assump:2}
    \int^{t+T_1}_t \Psi(\tau,u)\Psi^T(\tau,u)d\tau\geq\mu I>0,\ \ \forall  t\geq 0.
\end{equation}
\end{itemize}
\end{prop}
\begin{proof}
Defining the vector $e_V^T=\begin{bmatrix}e_2&e_3\end{bmatrix}$ and
substituting $k_1$ and $k_2$ in (\ref{ineqn:k1k2}) into the system in (\ref{eqn:reduced order dynamics}),
it follows that
\begin{equation}\label{eqn:reduced order dynamics1}
\begin{split}
\dot{e}_V=-\frac{k}{L}
\begin{bmatrix}
u^2_1&u_1u_2\\u_1u_2&u^2_2
\end{bmatrix}
e_V=-\Psi(\tau,u)\Psi^T(\tau,u)e_V,
\end{split}
\end{equation}
Because the switch signals $u$ are generated by a simple Pulse-width modulation(PWM),
$\|\Psi(t,u)\|\leq \sqrt{\frac{k}{L}}\|u\|\leq\sqrt{2\frac{k}{L}}$
and $T_1$ can be chosen as one period of the switching sequence to verify
the condition in (\ref{assump:2}).
Given that the conditions (\ref{conditon:1}) and (\ref{assump:2}) hold,
it follows from \citep{loria2002uniform} that the reduced-order system (\ref{eqn:reduced order dynamics}) is exponentially stable.
Thus, Proposition \ref{prop:1} is proven.
\end{proof}
\begin{remark}\label{rmk:applicable}
The proposed observer (\ref{eqn:stw observer}) is applicable to all converters that fall under the class of systems represented by (\ref{eqn:hybrid dynamic1}). This class applies to a wide range of hybrid switched-affine multi-cell converter systems (see \rm{\cite{kouro2010recent}}).
\end{remark}
\section{SIMULATION RESULTS}
The performance of the proposed adaptive-gain SOSML observer was evaluated
through simulations. To demonstrate the improvement of the proposed
strategy, the results are compared with a Luenberger switched observer
given in \cite{Riedinger20101047}.
The simulation parameters are shown in Table \ref{table_parameter}.
Furthermore, the load resistance was varied up to $50\%$ to demonstrate the robustness of the proposed observer.
\begin{table}[h]
\begin{center}
\caption{Main Parameters of Simulation Model}\label{table_parameter}
\begin{tabular}{cc}
\hline
System Parameters & Values\\
\hline
DC voltage($E$) & 150 $V$\\
\hline
Capacitors($c_1,c_2$) & 40 $\mu F$\\
\hline
Load resistance($R$) & 131 $\Omega$\\
\hline
Load Inductor($L$) &10 $mH$\\
\hline
The chopping frequency & 5 $kHz$\\
\hline
The sampling period & 5 $\mu s$\\
\hline
\end{tabular}
\end{center}
\end{table}

The system in (\ref{eqn:new system dynamic}) is rewritten in a form convenient for designing the Luenberger switched observer \citep{Riedinger20101047}:
\begin{equation}\label{eqn:Luenberger observer}
\left\{
\begin{split}
&\dot{\hat{I}}=-\frac{R}{L}\hat{I} + \frac{E}{L}u_3 - \frac{\hat{V}_{c_1}}{L}u_1-
\frac{\hat{V}_{c_2}}{L}u_2 + \kappa_0e_1,\\
&\dot{\hat{V}}_{c_1}=\frac{u_1}{c_1}I + (\kappa_1u_1+\kappa_3u_2+\kappa_5u_3)e_1,\\
&\dot{\hat{V}}_{c_2}=\frac{u_2}{c_2}I + (\kappa_2u_1+\kappa_4u_2+\kappa_6u_3)e_1.
\end{split}
\right.
\end{equation}
The error dynamics of $e^T=\begin{bmatrix}e_1&e_2&e_3\end{bmatrix}$
are given by equations (\ref{eqn:new system dynamic})
and (\ref{eqn:Luenberger observer}),
\begin{equation}\label{eqn:error Luenberger observer}
\begin{split}
\dot{e}=(\tilde{A}_0+u_1\tilde{A}_1+u_2\tilde{A}_2+u_3\tilde{A}_3)e,
\end{split}
\end{equation}
where $\tilde{A}_i=A_i-K_iC,\ \ i=0,1,2,3$,
$K^T_0=(\kappa_0,0,0),\ \ K^T_1=(0,\kappa_1,\kappa_2),\ \ K^T_2=(0,\kappa_3,\kappa_4),
\ \ K^T_3=(0,\kappa_5,\kappa_6)$.
The constant gains $\kappa_0,\kappa_1,\kappa_2,\kappa_3,\kappa_4,\kappa_5$ and $\kappa_6$ are chosen such that there exists a positive matrix $\tilde{P}$ that satisfies
$\tilde{A}^T_i\tilde{P}+\tilde{P}\tilde{A}_i\leq 0$, for $i=0,1,2,3$.
All the details of the parameters can be found in \cite{Riedinger20101047}.

For simulation purposes, the initial values were chosen as
\begin{equation}\label{eqn:initial condition}
    {V}_{c_1}(0)=5V,\ \ {V}_{c_2}(0)=10V,\ \ \hat{V}_{c_1}(0)=0V, \ \ \hat{V}_{c_2}(0)=0V.
\end{equation}
The parameters of the adaptive SOSML algorithm given by
(\ref{eqn:adaptive gains}) and (\ref{eqn:adaptive gains rt}) were chosen as
$\lambda_0=2,\ \ \alpha_0=4,\ \ k_{\lambda_0}=2.5,\ \ k_{\alpha_0}=20$ and
$k=6\cdot 10^5$.
The parameter of the switching gains in (\ref{ineqn:k1k2}) was chosen to be $\kappa=20$.
The inputs of the switches $u$ were generated by a simple PWM with
a chopping frequency 5 $kHz$, the sampling period was 200 $kHz$.

Figure \ref{vc12nonoise} shows the estimates of the capacitor voltages
$V_{c_1}$ and $V_{c_2}$ and the errors when the system is not affected by output noise and without load variation.
Both the adaptive-gain SOSML observer
and the Luenberger switched observer can achieve desired performance.

The estimates of the capacitor voltages
$V_{c_1}, V_{c_2}$ and the errors
when the system is affected by the output noise
and under load variations up to $50\%$ are shown in Figure \ref{vc12noise}.
The system output noise was included to test the robustness
of the proposed observers \citep{bejaranomulticell}, and this is shown in Figure \ref{fig:output noise}.
It is clear from the figures, the proposed observer is robust and the effect of the noise is essentially imperceptible. On the other hand, the
Luenberger switched observer is more sensitive to the noise and
the load variation.
From \cite{levent1998robust}, we know that the SOSM observer works
as a robust exact differentiator, and for this reason we obtain better performance
from the proposed observer compared with the Luenberger switched observer.
Figure \ref{Adaptive-gains} shows that the adaptive law of
(\ref{eqn:adaptive gains}) and (\ref{eqn:adaptive gains rt}) is effective
under load variations.
\begin{remark}\label{computational issues1}
From implementation point of view, the calculations required for the the adaptive-gain SOSML (\ref{eqn:stw observer}) are slightly more intensive than those of Leunberger observer. However, the correction term $\mu(e_1)$ and two design parameters $k_1, k_2$ entails low real-time computational burden, as the computational capabilities of digital computers have greatly increased and the additional processing requirements can be easily accomplished \rm{(see \cite{6153362,oettmeier2009mpc,lienhardt2007digital})}. As $\mu(e_1)$ is calculated only once, regardless of the number of cells, the complexity of the calculation increases linearly with the number of cells. This means that, for an $n$-cell system with $2^n$ permutations $(n>3)$, the additional computational burden comes only from the calculation of new parameters $k_3,\cdots,k_{n-1}$.
\end{remark}
\section{CONCLUSIONS}
In this paper, a novel adaptive-gain SOSML observer was presented
for a multi-cell power converter system, which belongs to a class of hybrid systems.
With the use of $Z(T_N)$-observability, the capacitor voltages were estimated
under a certain condition of the input sequences, even though the system did not satisfy the observability matrix rank condition. That is, the system becomes observable in the sense of $Z(T_N)$-observability after several switching sequences.
The robustness of the proposed observer and the Luenberger switched observer
were compared in the presence of load resistance variations
and output measurement noise.
It was found that the adaptive-gain SOSML observer was more robust than
the Luenberger switched observer.
Two main advantages of the proposed method are:
1) Only one parameter $k$ has to be tuned;
2) A-priori knowledge of the perturbation bounds is not required.
\begin{figure}[H]
\begin{center}
\subfigure[Estimate of $V_{c_1}$]{\label{fig:vc1_no_noise}
\includegraphics[width=2.2in]{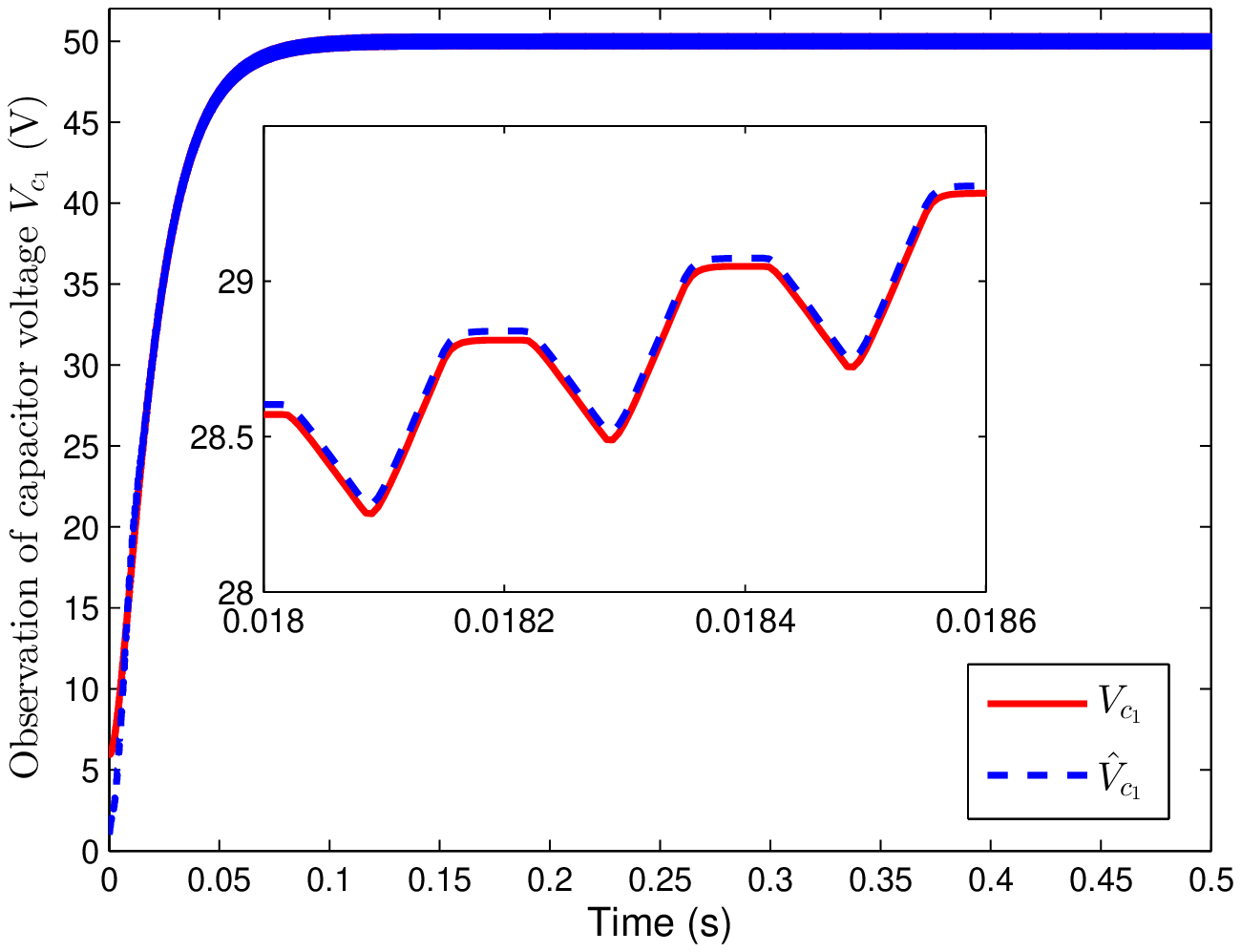}}
\quad
\subfigure[Estimate of $V_{c_1}$]{\label{fig:vc1_no_noise_luber}
\includegraphics[width=2.2in]{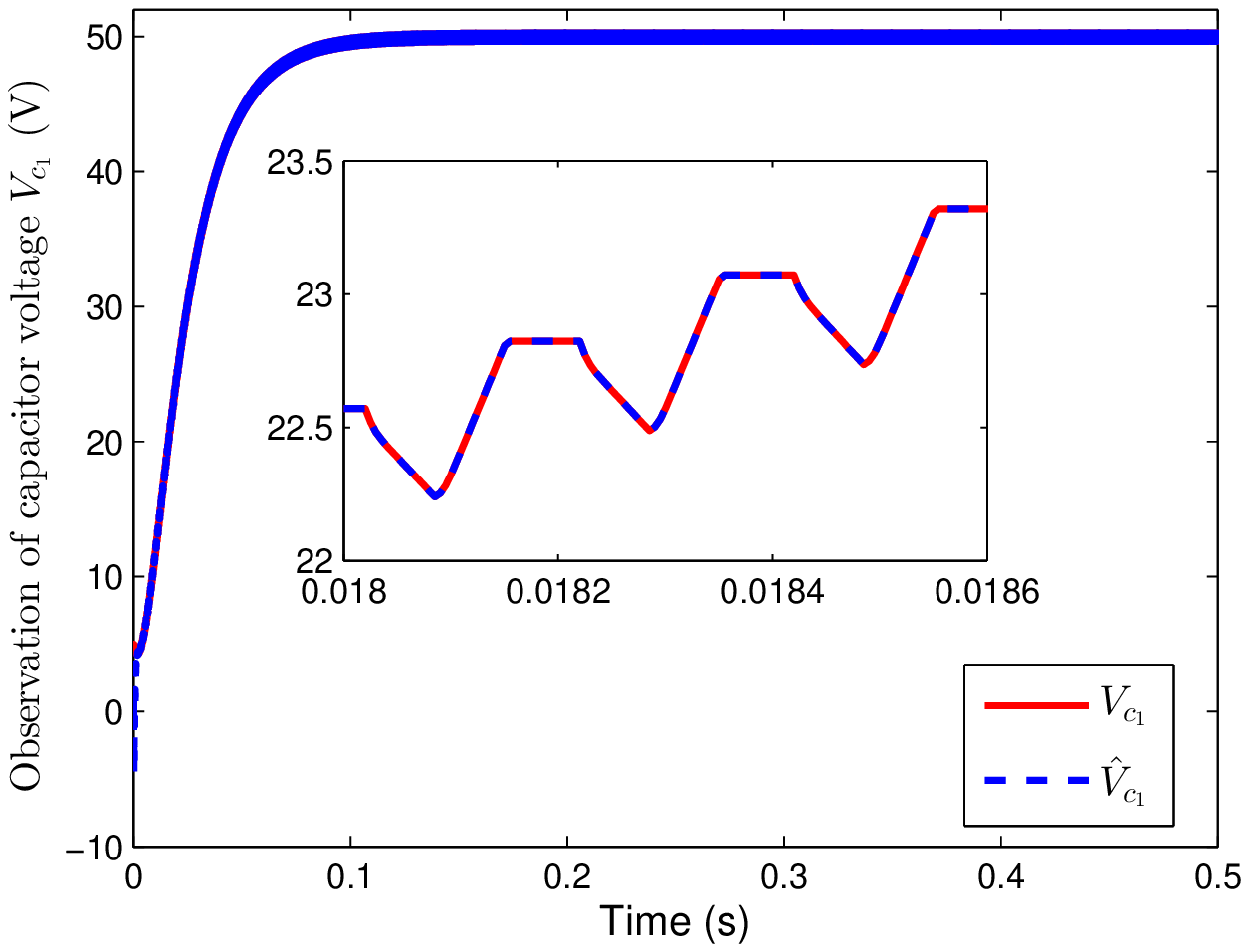}}\\
\subfigure[Estimate of $V_{c_2}$]{\label{fig:vc2_no_noise}
\includegraphics[width=2.2in]{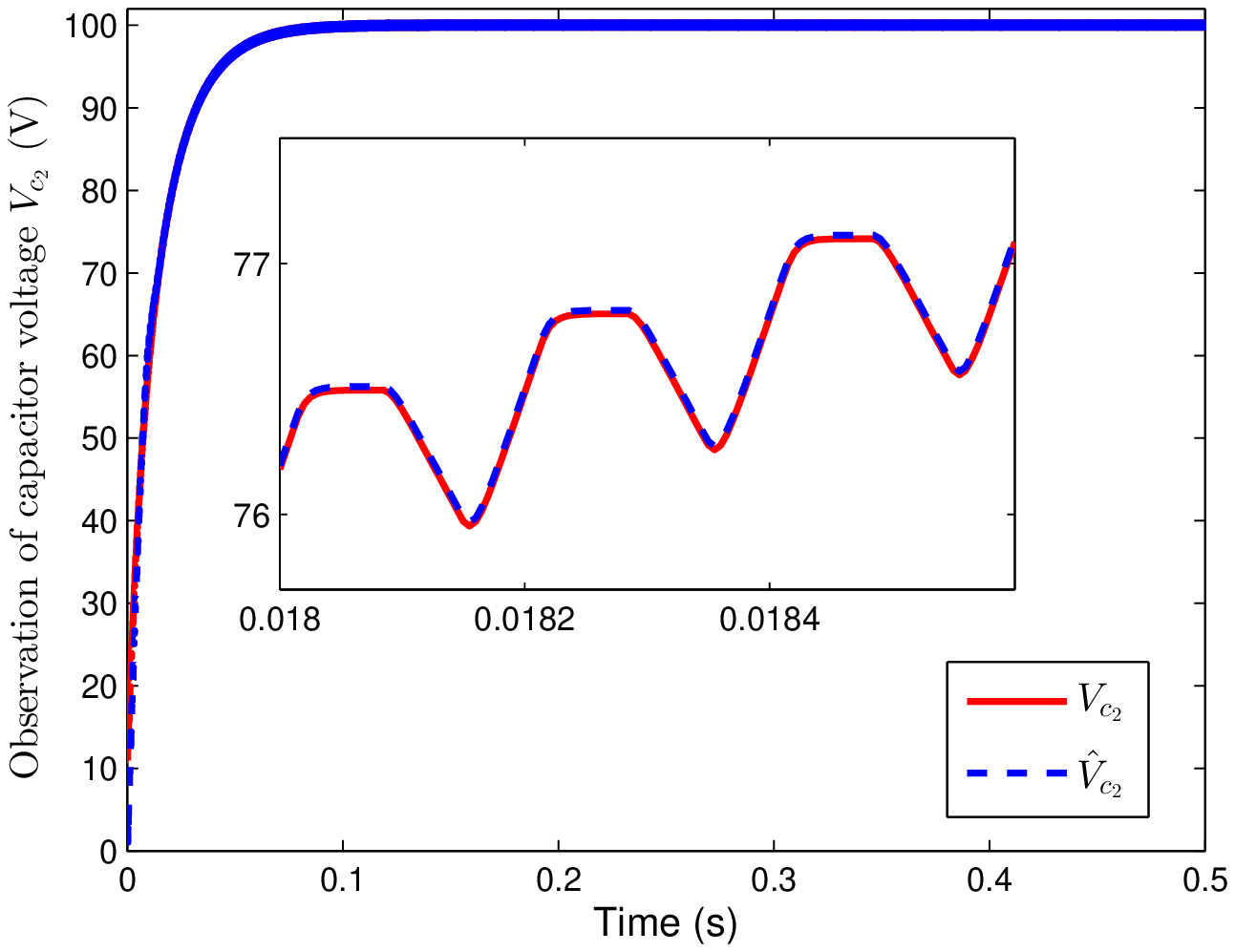}}
\quad
\subfigure[Estimate of $V_{c_2}$]{\label{fig:vc2_no_noise_luber}
\includegraphics[width=2.2in]{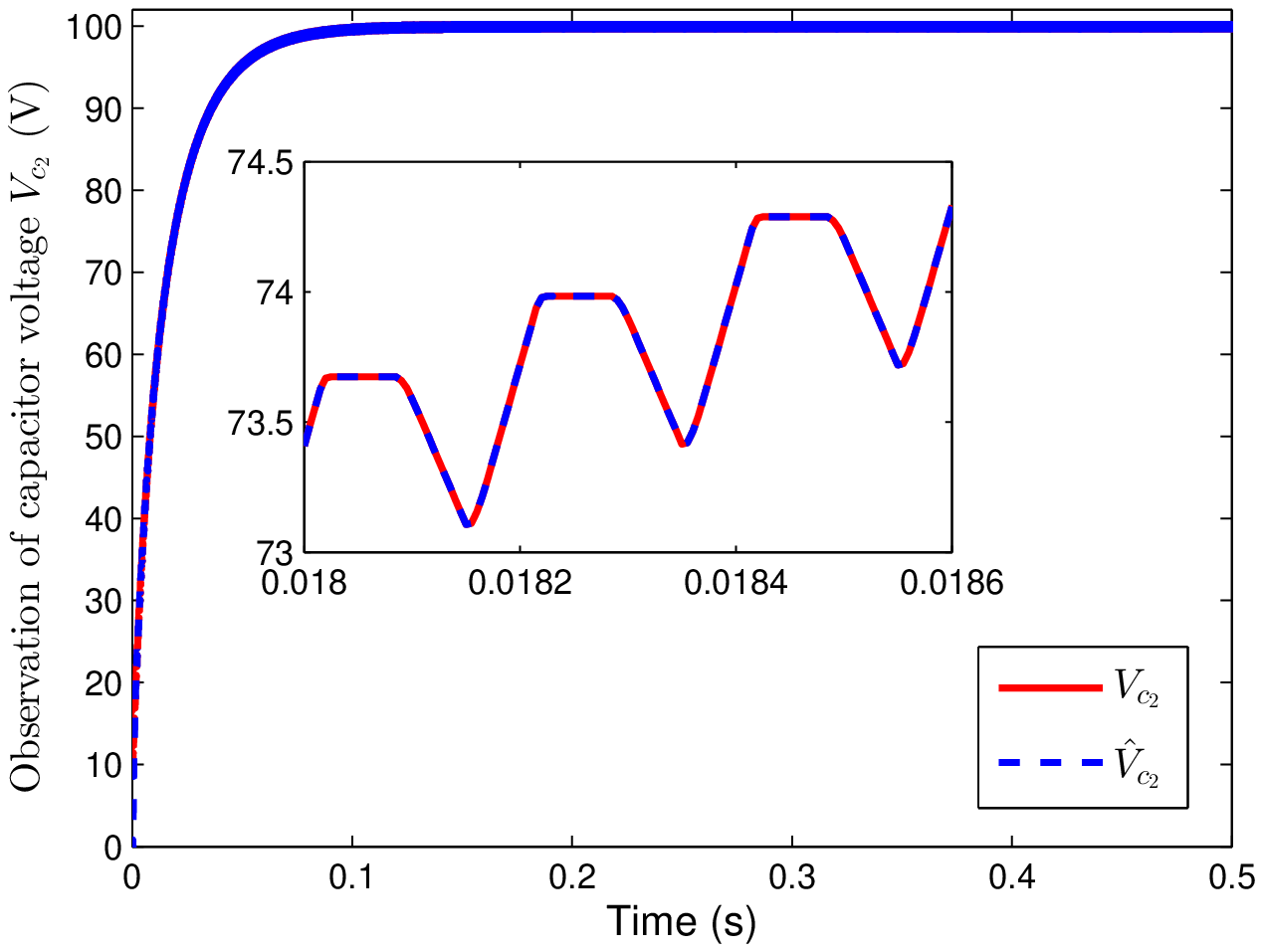}}\\
\subfigure[Estimation errors $e_{V_{c_1}},e_{V_{c_2}}$]{\label{fig:vc12_err_no_noise}
\includegraphics[width=2.2in]{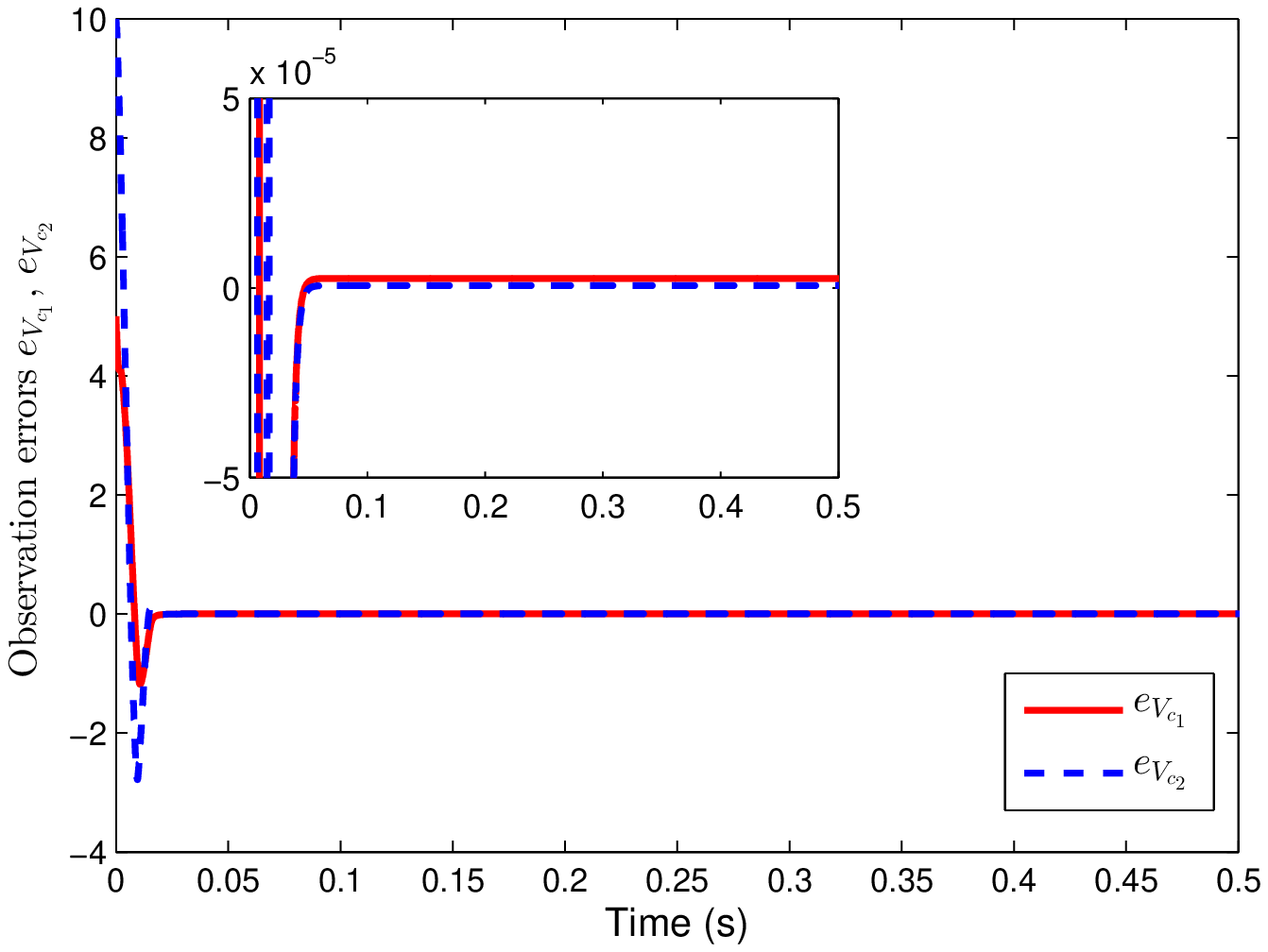}}
\quad
\subfigure[Estimation errors $e_{V_{c_1}},e_{V_{c_2}}$]{\label{fig:vc12_err_no_noise_luber}
\includegraphics[width=2.2in]{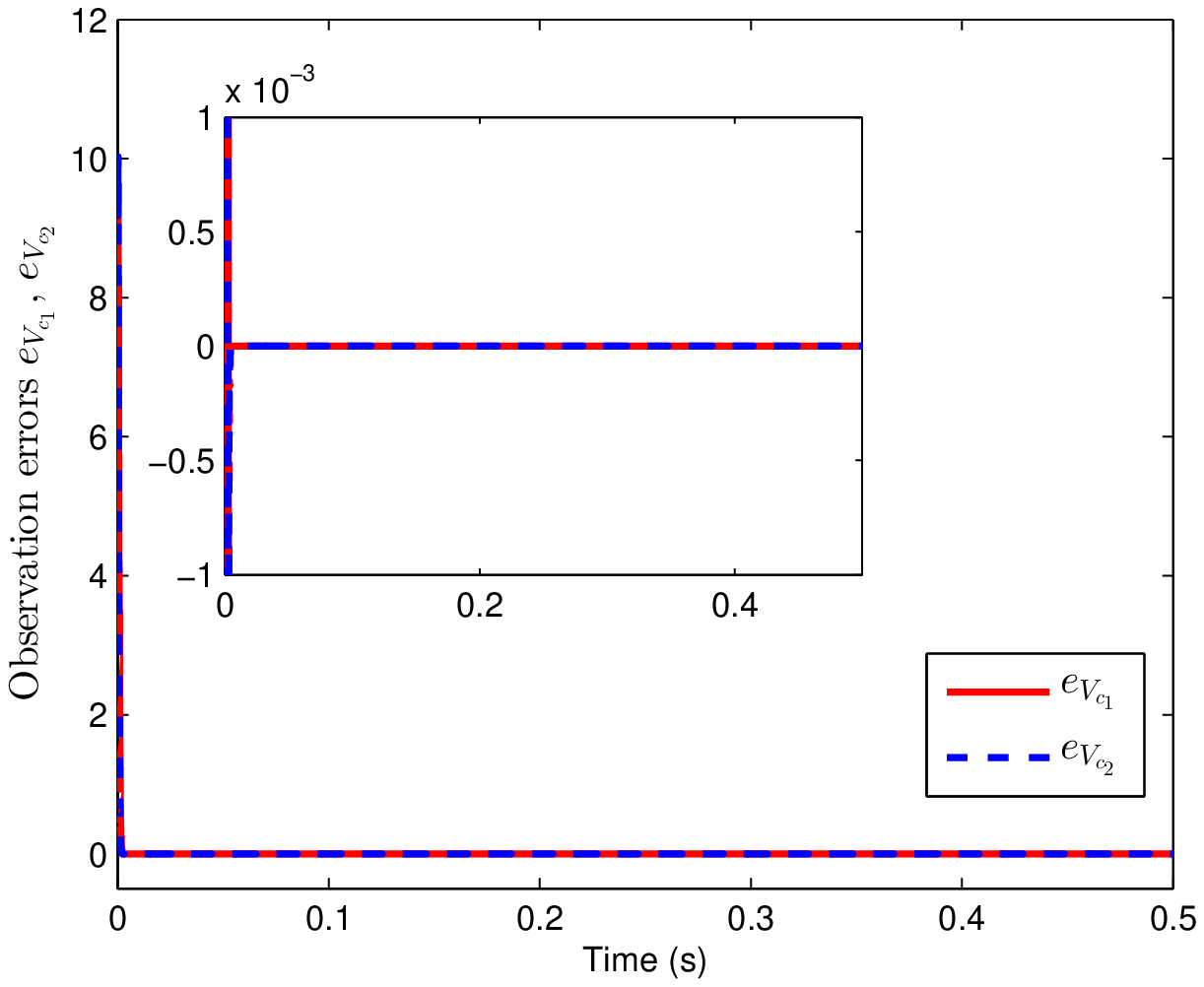}}
\caption{Estimate of capacitor voltage $V_{c_1},V_{c_2}$ and its errors
for adaptive-gain SOSML (\ref{fig:vc1_no_noise},
\ref{fig:vc2_no_noise}, \ref{fig:vc12_err_no_noise})
and Luenberger switched observer (\ref{fig:vc1_no_noise_luber},
\ref{fig:vc2_no_noise_luber}, \ref{fig:vc12_err_no_noise_luber}), respectively,
when the system output is not affected by noise and load variations}
\label{vc12nonoise}
\end{center}
\end{figure}
\begin{figure}[H]
\begin{center}
\subfigure[Estimate of $V_{c_1}$]{\label{fig:vc1_noise}
\includegraphics[width=2.2in]{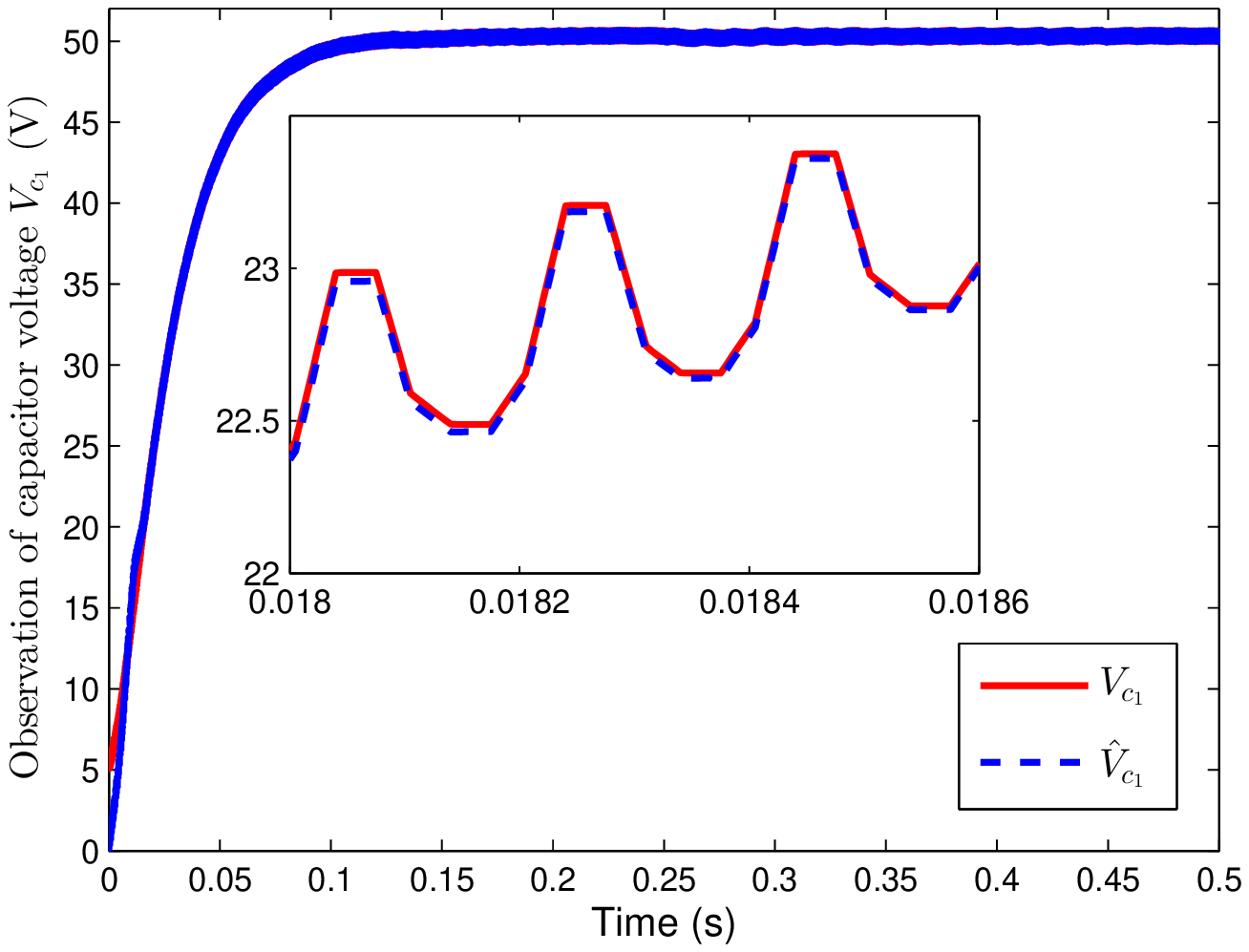}}
\quad
\subfigure[Estimate of $V_{c_1}$]{\label{fig:vc1_noise_luber}
\includegraphics[width=2.2in]{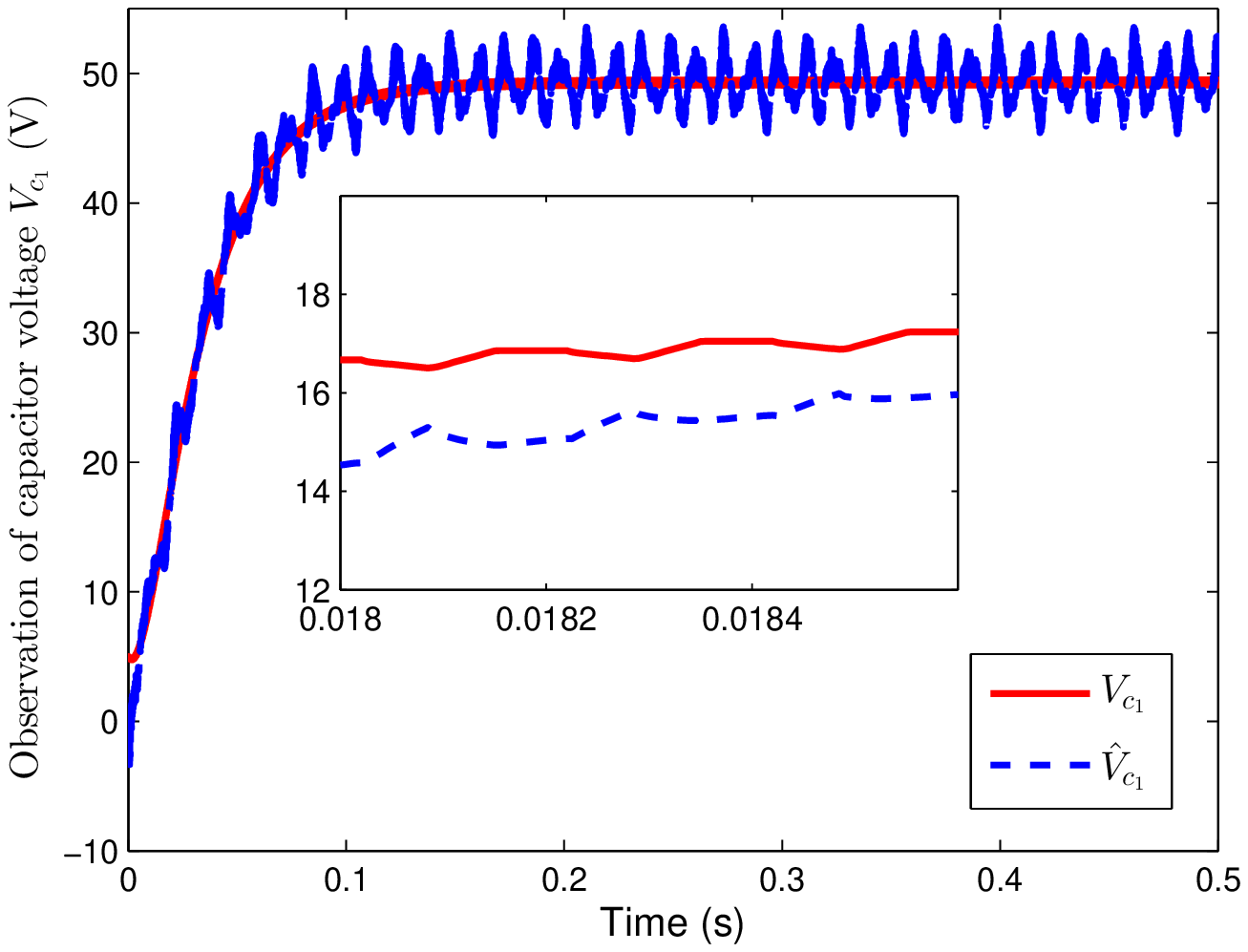}}\\
\subfigure[Estimate of $V_{c_2}$]{\label{fig:vc2_noise}
\includegraphics[width=2.2in]{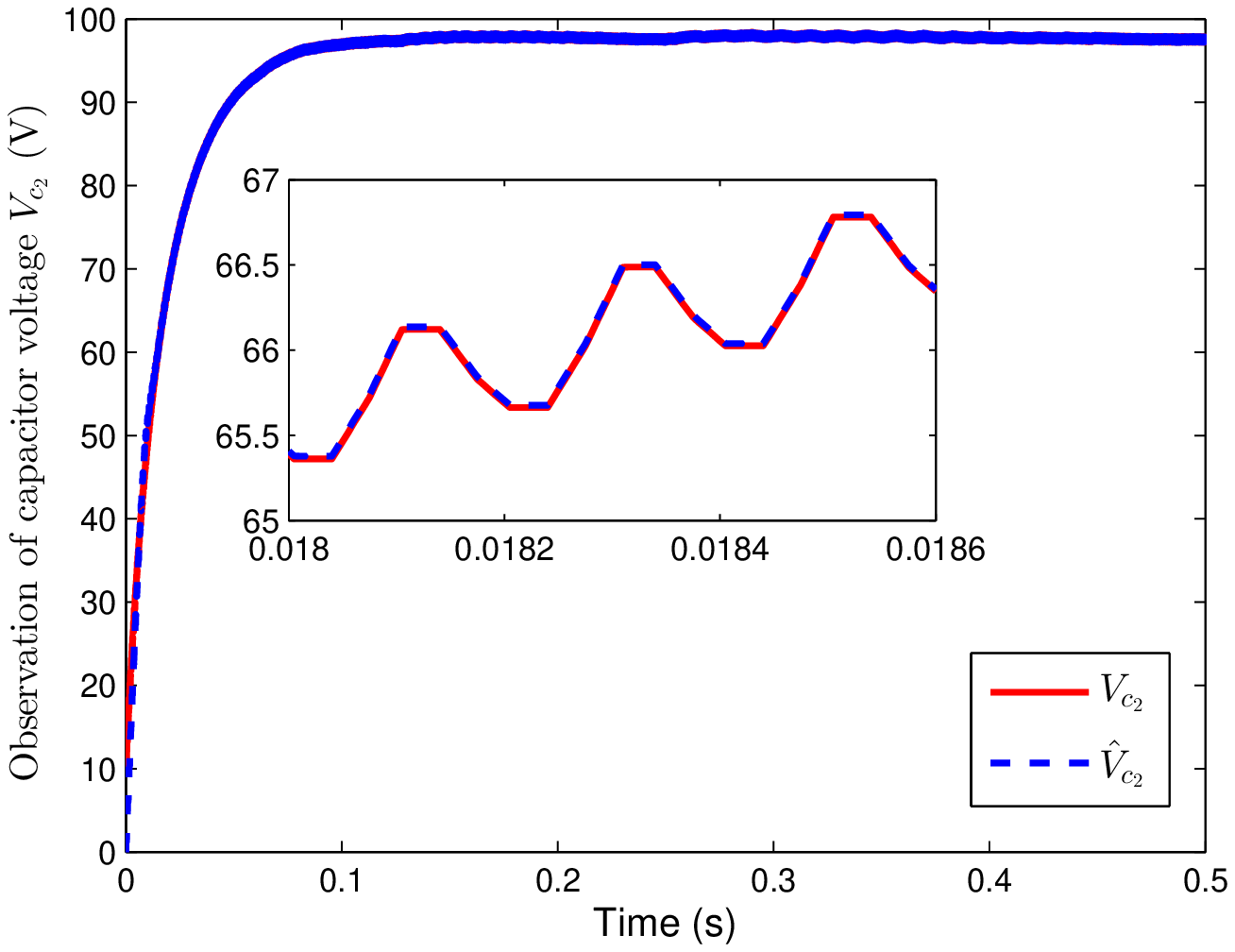}}
\quad
\subfigure[Estimate of $V_{c_2}$]{\label{fig:vc2_noise_luber}
\includegraphics[width=2.2in]{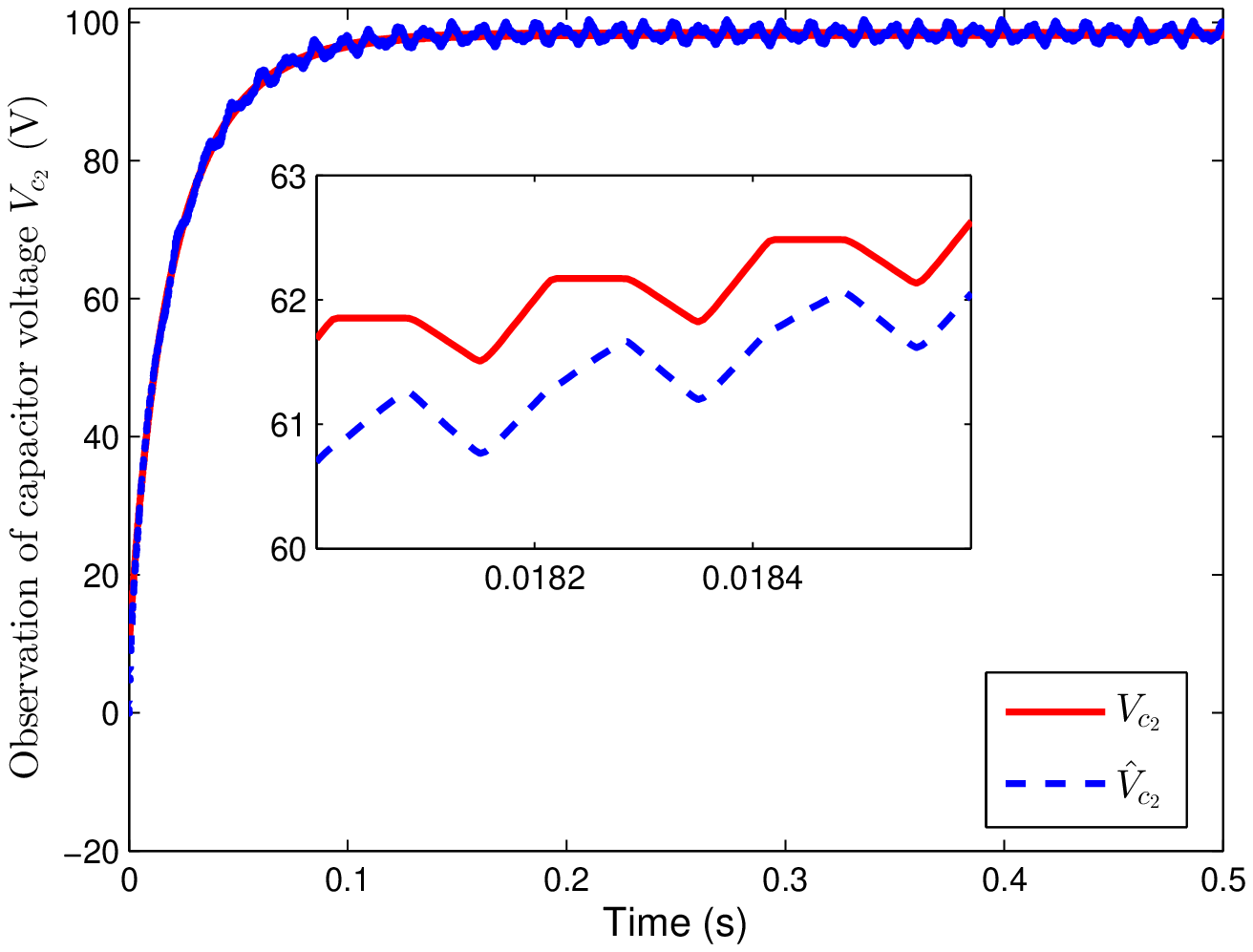}}\\
\subfigure[Estimation errors $e_{V_{c_1}},e_{V_{c_2}}$]{\label{fig:vc12_err_noise}
\includegraphics[width=2.2in]{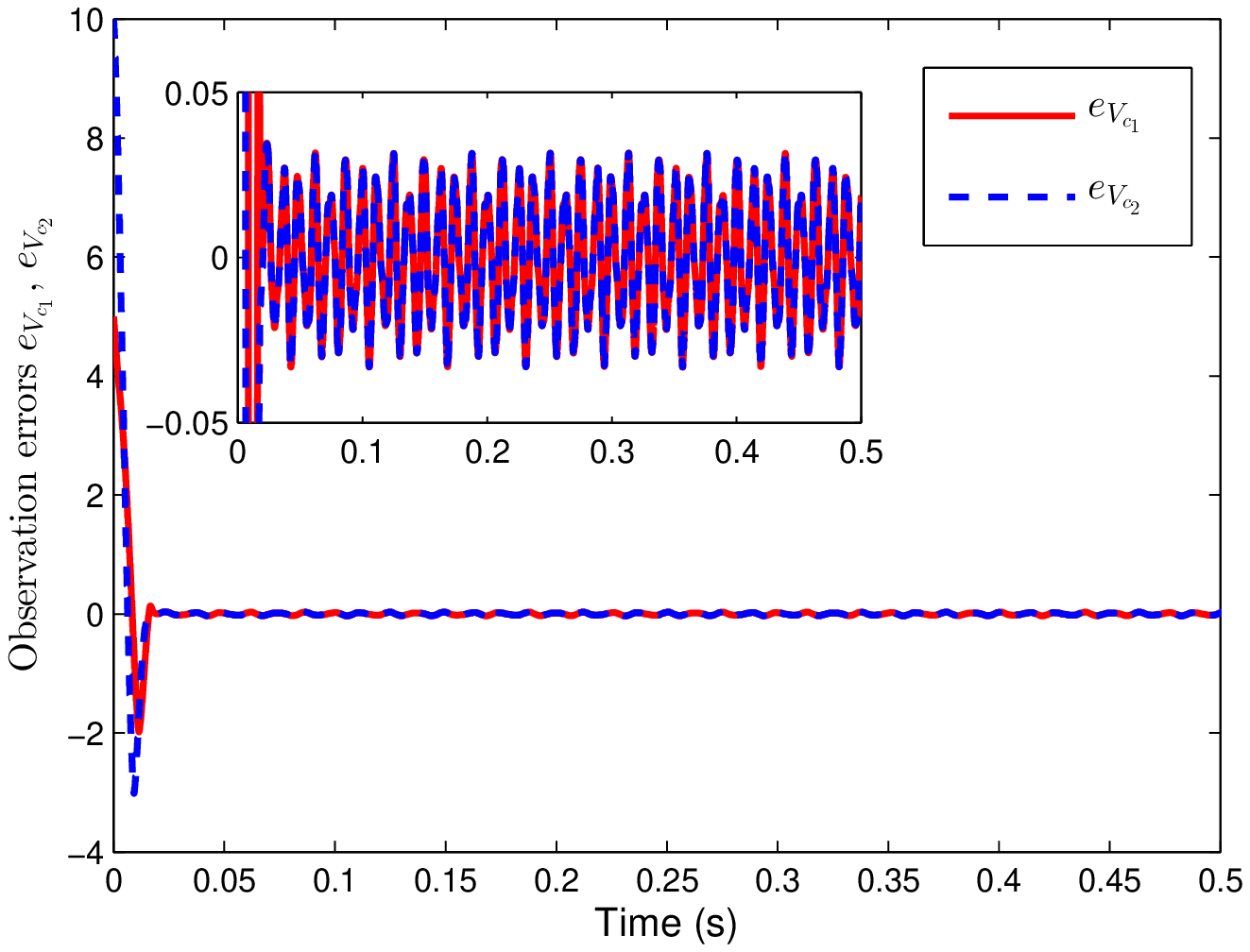}}
\quad
\subfigure[Estimation errors $e_{V_{c_1}},e_{V_{c_2}}$]{\label{fig:vc12_err_noise_luber}
\includegraphics[width=2.2in]{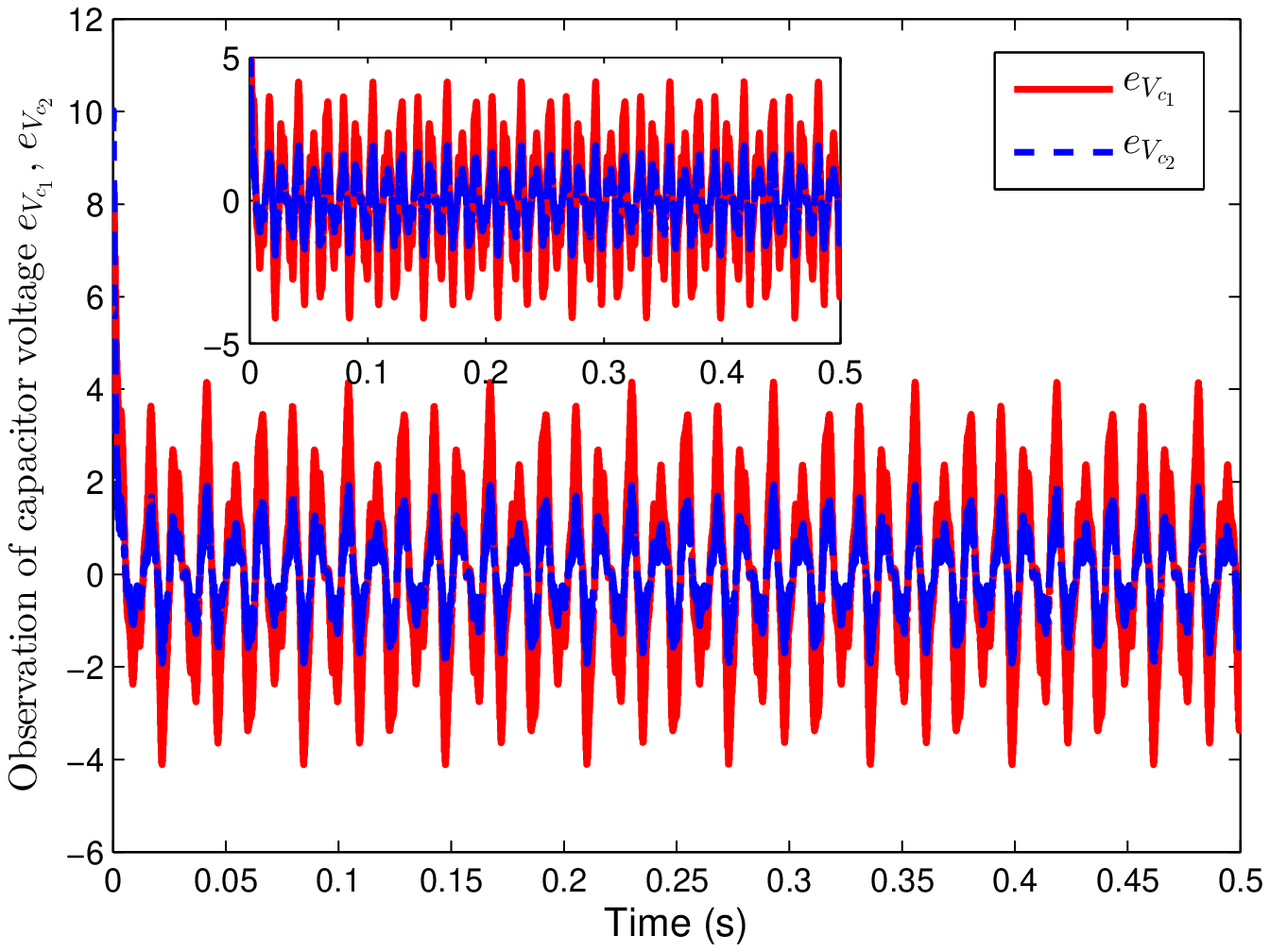}}
\caption{Estimate of capacitor voltage $V_{c_1},V_{c_2}$ and its errors
for adaptive-gain SOSML (\ref{fig:vc1_noise},
\ref{fig:vc2_noise}, \ref{fig:vc12_err_noise})
and Luenberger switched observer (\ref{fig:vc1_noise_luber},
\ref{fig:vc2_noise_luber}, \ref{fig:vc12_err_noise_luber}), respectively,
when the system output is affected by noise
and load variations $\hat{R}=1.5R$}
\label{vc12noise}
\end{center}
\end{figure}
\begin{figure}[H]
\begin{center}
\includegraphics[width=8.4cm]{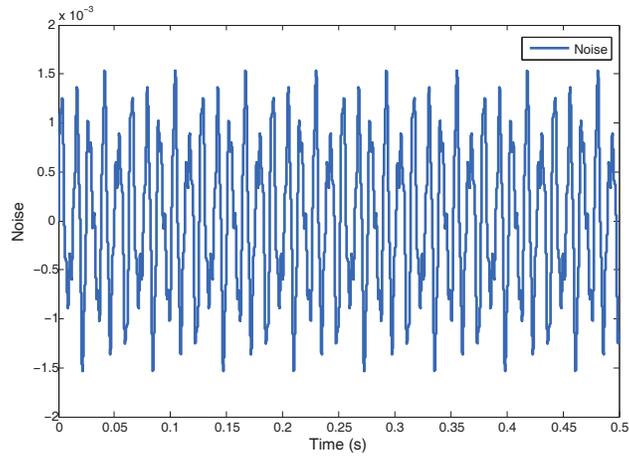}
\caption{System output noise}
\label{fig:output noise}
\end{center}
\end{figure}
\begin{figure}[H]
\begin{center}
\includegraphics[width=8.4cm]{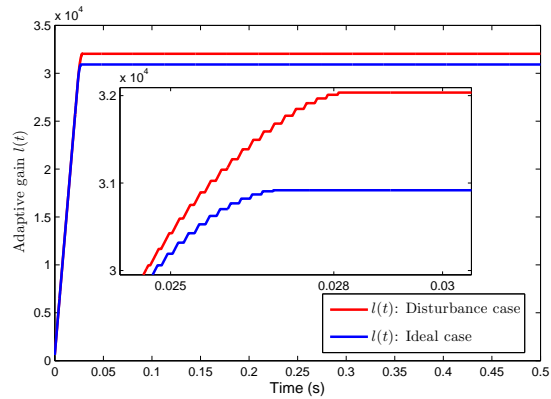}
\caption{Adaptive law $l(t)$ of the SOSML algorithm}
\label{Adaptive-gains}
\end{center}
\end{figure}
\bibliographystyle{model5-names}
\bibliography{multicell_CEP}
\end{document}